\documentclass[journal]{IEEEtran}
\usepackage{url}            
\usepackage{booktabs}       
\usepackage{amsfonts}       
\usepackage{nicefrac}       
\usepackage{microtype}      
\usepackage{lipsum}
\usepackage{amsmath}
\usepackage{graphicx}
\usepackage{graphicx}
\usepackage{amssymb}
\usepackage{pdfpages}
\usepackage{float}
\usepackage{subfiles} 
\usepackage{amsthm}
\usepackage{mathtools}
\usepackage{mathrsfs}
\usepackage{subfig} 
\newtheorem{theorem}{Theorem}

\newtheorem{definition}{Definition}
\newtheorem{remark}{Remark}

\newtheorem{assumption}{Assumption}
\newtheorem{problem}{Problem}
\DeclarePairedDelimiter{\diagfences}{(}{)}
\newcommand{\diag}{\operatorname{diag}\diagfences}

\usepackage{algorithm}
\usepackage{algpseudocode}
\newcommand\ignore[1]{{}}
\algrenewcommand\algorithmicrequire{\textbf{Input:}}
\algrenewcommand\algorithmicensure{\textbf{Output:}}

\begin{document}

\author{Cong Li,~\IEEEmembership{}
        Yongchao Wang, ~\IEEEmembership{} Fangzhou Liu{$^*$}, ~\IEEEmembership{}
        and~Martin Buss
\thanks{C. Li, Y. Wang, F. Liu, and M.Buss are with the Chair of Automatic Control Engineering, Technical University of Munich, Theresienstr. 90,80333, Munich, Germany e-mail: \{cong.li, yongchao.wang, fangzhou.liu, mb\}@tum.de.}
}
\title{Off Policy Risk Sensitive Reinforcement Learning Based Optimal Tracking Control with Prescribe Performances}
\maketitle

\begin{abstract}
    An off policy reinforcement learning based control strategy is developed for the optimal tracking control problem to achieve the prescribed performance of full states during the learning process. The optimal tracking control problem is converted as an optimal regulation problem based on an auxiliary system. The requirements of prescribed performances are transformed into constraint satisfaction problems that are dealt by risk sensitive state penalty terms under an optimization framework. To get approximated solutions of the Hamilton Jacobi Bellman equation, an off policy adaptive critic learning architecture is developed by using current data and experience data together. By using experience data, the proposed weight estimation update law of the critic learning agent guarantees weight convergence to the actual value. This technique enjoys practicability comparing with common methods that need to incorporate external signals to satisfy the persistence of excitation condition for weight convergence. The proofs of stability and weight convergence of the closed loop system are provided.
    Simulation results reveal the validity of the proposed off policy risk sensitive reinforcement learning based control strategy.
\end{abstract}

\IEEEpeerreviewmaketitle

\section{Introduction}
\IEEEPARstart{T}{he} optimal tracking control problem (OTCP) has always been the focus of the control community, in which both tracking error variations and control energy expenditures serve as performance indexes to be optimized. (see, e.g. \cite{bertsekas1995dynamic,lewis2012optimal}, and the references therein). The control strategy for the OTCP is derived from solving an algebraic Riccati equation (ARE) of a linear system, or a Hamilton Jacobi Bellman (HJB) equation of a nonlinear system. However, it is well known that it is not easy to solve the ARE or HJB directly. In recent days, adaptive (approximate) dynamic programming (ADP) emerges as an efficient reinforcement learning (RL) based framework to get approximated solutions of the ARE or HJB equation based on an actor-critic artificial neural network (NN) approximation scheme. 
Despite successful applications of the ADP learning framework for the OTCP, the performance guarantee of full states during the whole learning process has yet to be established. In ADP related works, although the tracking error over a long horizon is minimized during the learning process, the instantaneous tracking performance is ignored and may cause safety issues. For example, the weight adaption in the initial training period might generate a harmful overshoot that is higher than an acceptable threshold, and it may lead to a loss of stability or damage to hardwares. 
Besides, the existing OTC related works mainly achieve the tracking error convergence to an uncertain residual set, whose size relies on hyper-parameters chosen. For a given task with predefined requirements of the tracking error, we prefer a quantized guarantee for the final achievable tracking performance.
Taking the above into consideration, we aim to develop an effective control strategy to solve the OTCP while guarantees the tracking performance during the whole learning process.
\subsection{Prior and related works}
Considering the tracking error for the OTCP, prescribed performance functions (PPFs) are firstly proposed in \cite{bechlioulis2010prescribed,bechlioulis2008robust} to guarantee that the tracking error converges to an arbitrarily small residual set, convergence rate is no less than a predefined value, and a maximum overshoot is less than a prespecified constant. Then, a PPF based system transformation method is often combined with the backstepping technique to achieve desired performances \cite{guo2018neural,huang2018neuro}. Besides, PPFs serve to construct barrier Lyapunov functions (BLFs) to enforce satisfaction of prescribed performances (PPs) under a recursive controller design process \cite{yang2020adaptive}.  
However, control energy expenditures are not considered in these works. 
Later, PPFs are incorporated into an optimization framework to consider performance criteria in terms of both tracking errors and energy expenditures \cite{dong2020optimal,wang2017dynamic}. However, to the best of our knowledge, most of PP related works under an optimization framework only focus on the strict-feedback system or the pure-feedback system, and just achieve PPs of output states based on a PPF guided system transformation technique. 
For certain practical applications, it is desirable to guarantee performance for full states of the investigated system. 
For example, we prefer a robot manipulator to track the reference trajectory precisely (angular position errors) and smoothly (angular velocity errors).
\subsection{Contribution}
Our work builds on the problem transformation method illustrated in \cite{kamalapurkar2015approximate}, where the OTCP of the investigated system is converted into an equivalent stationary optimal regulation problem of an auxiliary system. 
The PPs of full states are interpreted as  tracking error constraints, and risk sensitive state penalty (RS-SP) terms from our prior work \cite{li2020online} are developed to tackle these constraints. 
The resulting PPF based RS-SP terms are incorporated into the cost function to transform the constrained optimization problem into an unconstrained optimization problem. Based on the off policy adaptive critic learning architecture, the approximated optimal control strategy derived by solving the unconstrained optimal regulation problem achieves the OTC with PPs.
Comparing with existing works, the contributions are summarized as follows:
(a) The proposed RS-SP terms based optimization framework guarantees PP for full states of a general nonlinear system;
(b) The off policy RL based control strategy developed in our prior work \cite{li2020online} is applied in the trajectory tracking scenario to test its effectiveness. 
The avoiding of incorporating external signals to satisfy the persistence of excitation (PE) condition to achieve the weight convergence enables the satisfaction of PP of full states feasible. 
Otherwise, in the initial learning period, the real trajectory may in a random form under the influence of external signals, which may put the robot manipulator at an unsafe state.
\subsection{Paper organization}
Section 2 introduces preliminaries, the problem formulation of the prescribed performance optimal tracking control problem, and the problem transformation.
Section 3 briefly elucidates the off policy RL based control strategy to solve the transformed risk sensitive optimal regulation problem, followed by simulation results shown in Section 4 to test effectiveness of the proposed strategy. Conclusions are provided in Section 5.

\emph{Notations:} Throughout this paper, $\mathbb{R}_{+}$ denotes the set of positive real numbers; 
$\mathbb{R}^{n}$ is the Euclidean space of $n$-dimensional real vector; 
$\mathbb{R}^{n \times m}$ is the Euclidean space of $n \times m$ real matrices; 
$I_{m \times m}$ represents the identity matrix with dimension $m \times m$; $0_{n \times m}$ denotes the $n \times m$ zero matrix; 
$\lambda_{\min}(M)$ and $\lambda_{\max}(M)$ are the maximum and minimum eigenvalues of a symmetric matrix $M$, respectively; 
$\diag {a_{1},...,a_{n}}$ is the $n \times n$ diagonal matrix with the value of main diagonal as $a_{1},...,a_{n}$.
The $i$th entry of a vector $x = [x_{1},...,x_{n}]^{\top}\in \mathbb{R}^{n}$ is denoted by $x_{i}$, and $\left\| x \right\| = \sqrt{\sum_{i=1}^{N}|x_{i}|^2}$ is the Euclidean norm of the vector $x$.
The $ij$th entry of a matrix $A \in \mathbb{R}^{m \times n}$ is denoted by $a_{ij}$, and $\left\|A\right\| = \sqrt{\sum_{i=1}^{m}\sum_{j=1}^{n}|a_{ij}|^2}$ is the Frobenius norm of the matrix $A$. For notational brevity, time-dependence is suppressed without causing ambiguity. 
\section{Preliminaries and problem formulation}
Consider the following general nonlinear dynamics:
\begin{equation} \label{original sys}
    \dot{x} = f(x)+g(x)u(x),
\end{equation}
where $x \in \mathbb{R}^{n}$ and $u(x) \in \mathbb{R}^{m}$ are states and inputs of the system. $f(x) : \mathbb{R}^{n} \to \mathbb{R}^{n}$, $g(x) : \mathbb{R}^{n} \to \mathbb{R}^{n \times m}$ are the known drift dynamics and input dynamics, respectively. 
\begin{assumption}\cite{kamalapurkar2015approximate} \label{bound of f g}
    The drift dynamics $f$ is Lipschitz continuous and $f(0) = 0$. 
    The input dynamics $g$ is bounded, and its inverse function $g^{+}=(g^{\top}g)^{-1}g^{\top} \in \mathbb{R}^{m}$ is bounded and Lipschitz continuous.
\end{assumption}
The control objective of this paper is to track the reference trajectory $x_r \in \mathbb{R}^{n}$ while the tracking error $e = x-x_r \in \mathbb{R}^{n}$ is guaranteed to satisfy the predefined performance criteria in terms of convergence rate, maximum overshoot, and residual set. These requirements for the tracking error can be reflected by the PPF that is defined as follows.
\begin{definition} [Prescribed performance function]  \cite{bechlioulis2008robust} \label{defination PF}
A smooth function $\rho: \mathbb{R}_{+} \to \mathbb{R}_{+}$ is called a prescribed performance function if: $\rho(t)$ is positive and decreasing, and $\lim_{t \to \infty} \rho(t) = \rho_{\infty} > 0$.
\end{definition}
To satisfy the above illustrated tracking performance, a PPF from \cite{bechlioulis2008robust} is adopted here.
\begin{equation} \label{PF example}
 \rho(t) = (\rho_0 - \rho_{\infty}) e^{-lt}+\rho_{\infty}.
\end{equation}
Based on the given PPF \eqref{PF example}, the goal of this paper is formulated as Problem \ref{PPTC}.
\begin{problem} [Prescribed performance optimal tracking control problem (PP-OTCP)] \label{PPTC}
   Given Assumption \ref{bound of f g}, design a control strategy $u(x)$ for the dynamics \eqref{original sys} to track the reference trajectory $x_r$ precisely. The control energy is minimized and the tracking error $e$ satisfies the desired performance as
    \begin{equation} \label{Performance function}
    -\alpha_i \rho_i (t) < e_i(t) < \alpha_i \rho_i (t), i = 1,\dots,n,
    \end{equation}
    where $\rho_i$ is the $i$th PPF in \eqref{PF example} that associates with the $i$th element of the tracking error $e_i$, the constant $\alpha_i \in \mathbb{R}_{+}$ is adjusted to denote desired performances.
\end{problem}
An intuitive explanation of \eqref{Performance function} is displayed in Fig.\ref{fig PF illustration}.
Assuming that $\left|e(0)\right|$ lies in the scope $[0,\alpha \rho_{0}]$, the final maximum residual set is represented by $[-\alpha \rho_{\infty},\alpha \rho_{\infty}]$, the maximum overshoot is confined to the scope $[- \alpha \rho_{0},\alpha \rho_{0}]$, and the convergence speed of $e(t)$ relates with the decreasing rate of $\rho(t)$ that is determined by the value of $l$. It should be clear that the parameter selection for \eqref{PF example} is determined by considering requirements of performance and safety together. 
For example, the determination of the scope $[-\alpha \rho_{0},\alpha \rho_{0}]$ for $e(t)$ should also take the limited working space into consideration.
\begin{figure}[!t]
    \centering
    \includegraphics[width=3.8in]{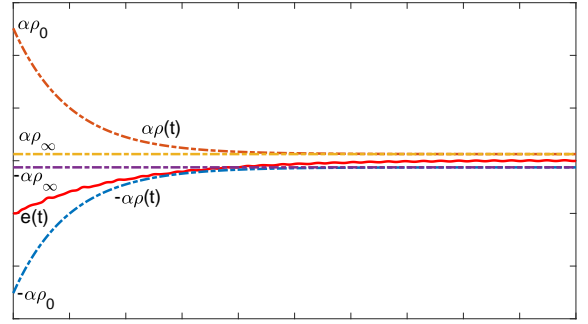}
    \caption{Graphical illustration of the relationship between the tracking error $e(t)$ and the PPF $\rho(t)$.}
    \label{fig PF illustration}
\end{figure}
\subsection{Problem transformation} \label{Problem trans}
To facilitate the problem transformation, the following assumption about the desired trajectory $y(x_r)$ is firstly given.
\begin{assumption}\cite{kamalapurkar2015approximate} \label{bound of xr}
    The reference trajectory follows $\dot{x}_r = y (x_r)$ where $y :\mathbb{R}^{n} \to \mathbb{R}^{n}$ is locally Lipschitz, and 
$x_r$ is bounded by $\left\| x_r \right\| \leq r \in \mathbb{R}_{+}$.
\end{assumption}
For the considered OTCP, a concatenated state $\eta\in \mathbb{R}^{2n}$ is introduced to transform it as an optimal regulation problem (ORP), which permits us to adopt RS-SP terms to tackle requirements of performances. The concatenated state $\eta$ is defined as 
\begin{equation} \label{concate state}
    \eta = [e,x_r]^{\top}.
\end{equation}
Calculating the time derivative of \eqref{concate state} yields
\begin{equation} \label{new system}
    \dot{\eta} = F(\eta) + G(\eta) \mu,
\end{equation}
where
$F(\eta) = \begin{bmatrix}
            f(e+x_r)-y(x_r)+g(e+x_r)\nu \\
             y_r(x_r) \\
        \end{bmatrix}: \mathbb{R}^{2n} \to \mathbb{R}^{2n}$,
$G(\eta) = \begin{bmatrix}
            g(e+x_r) \\
             0_{n \times m}  \\
        \end{bmatrix} : \mathbb{R}^{2n} \to \mathbb{R}^{2n \times m} $, and $\mu = u-\nu \in \mathbb{R}^{m}$.
The steady-state control policy $\nu$ with regard to the reference trajectory follows
\begin{equation} \label{control v PP}
    \nu = g^{+}(x_r) (y(x_r)-f(x_r)).
\end{equation}
Based on Assumption \ref{bound of f g}-\ref{bound of xr} and $f(0)=0$, we know that $F(0)=0$ and $F$ is locally Lipschitz. Based on the boundness of $g$, it is reasonable to conclude that $\left\|G\right\| \leq b_{G}$.

In order to solve Problem \ref{PPTC}, based on the auxiliary system \eqref{new system} and risk sensitive terms, the PP-OTCP in Problem \ref{PPTC} is equivalent to the risk sensitive optimal regulation problem shown in Problem \ref{ORP}.
\begin{problem} [risk sensitive optimal regulation problem (RS-ORP)] \label{ORP}
Given Assumption \ref{bound of f g}-\ref{bound of xr}, find a control policy $\mu(\eta)$ for the auxiliary system \eqref{new system} to minimize the following cost function
\begin{equation}\label{PP cost fuction 1}
    V(\eta) = \int_{t}^{\infty} r(\eta(\tau),\mu(\eta(\tau)))\,d\tau,
\end{equation}
where the constructive utility function $r(\eta,\mu(\eta)) = P(\eta)+\mu^{\top}R\mu$,
and the prescribed performance related penalty function is defined as 
\begin{equation}\label{performance penalty function}
    P(\eta) = \sum_{i=1}^{n} k_{i} \log \frac{\alpha^2_{i}}{\alpha^2_{i}-\zeta^2_{i}}+h_{i} \log \frac{\beta^2_{i}}{\beta^2_{i}-\delta^2_{i}},
\end{equation}
where $\zeta_i = e_i/\rho_i$, $\delta_i = x_{r_i}/\rho_i$, and $k_i$, $h_{i}$ are risk awareness parameters to be designed. 
\end{problem}
In Problem \ref{ORP}, the PPs for the tracking errors are interpreted as constraints that are tackled by PPF based RS-SP terms under an optimization framework. 
The working scheme of $ P(\eta)$ is displayed in Fig.\ref{PPFWorkingScheme}. Intuitively speaking,  $P(\eta)$ acts as barriers at the constraint boundaries defined by PPFs, and confines the tracking error remain in the region that satisfies desired performances.
\begin{figure}[!t]
    \centering
    \includegraphics[width=3.8in]{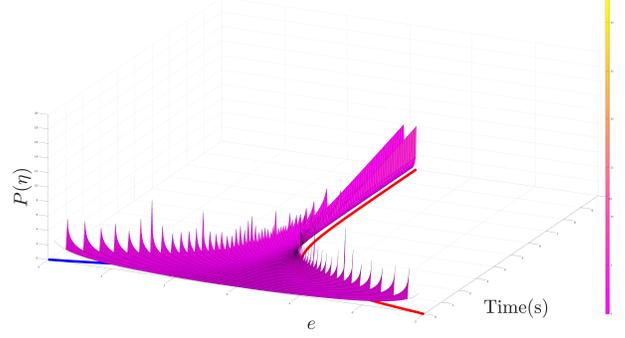}
    \caption{Graphical illustration of the working scheme of the PPF based RS-SP terms.}
    \label{PPFWorkingScheme}
\end{figure}
\begin{remark}
The construction of $P(\eta)$ is inspired by the barrier Lyapunov function (BLF) developed in \cite{tee2009barrier}. $P(\eta)$ is positive and approaches to infinity when the tracking error transgresses the boundaries of PPFs . The development of $P(\eta)$ allows us to consider the PPs for full states easily. 
It should be clear that the reference trajectory $x_r$ is also included into \eqref{performance penalty function}, which permits us to adopt the optimal value function as a Lyapunov function candidate to provide the stability proof. The incorporation of $x_r$ into \eqref{performance penalty function} results in an inevitable performance compromise problem, which can be dealt by setting the corresponding PPF as a loose one (a big $\beta_i$) with a sufficiently small $h_i$.
\end{remark}
\subsection{Hamilton Jacobi Bellman equation for auxiliary system}
Considering Problem \ref{ORP}, for any admissible control policies $\upsilon \in \Psi (\Omega)$ defined as \cite[Definition 1]{abu2005nearly}, the associated cost function is
    \begin{equation} \label{cost function for admissible policy}
    V^{\upsilon}(\eta) = \int_{t}^{\infty} r(\eta(\tau),\upsilon(\eta(\tau)))\,d\tau.
    \end{equation}
    Taking derivative of \eqref{cost function for admissible policy} with regard to $t$ yields the so-called nonlinear Lyapunov equation (LE),
    \begin{equation}\label{nonlinear Lyapunov equation}
    0 = r(\eta,\upsilon(\eta)) + \nabla V^{\top}(F(\eta)+G(\eta)\upsilon(\eta)),
    \end{equation}
    where the operation operator $\nabla$ denotes partial derivative with regard to $\eta$, i.e. $\partial / \partial \eta $.
    
    As for the RS-ORP in Problem \ref{ORP}, the associated optimal cost function is 
    \begin{equation} \label{optimal cost function}
    V^{*}(\eta) = \min_{\upsilon \in \Psi (\Omega)}\int_{t}^{\infty} r(\eta(\tau),\upsilon(\eta(\tau)))\,d\tau.
    \end{equation}
    Define the Hamiltonian as 
    \begin{equation}\label{Hamiltonian}
    \begin{aligned}
     H(\eta,\upsilon(\eta),\nabla V) = r(\eta,\upsilon(\eta)) + \nabla V^{T}(F(\eta)+G(\eta)\upsilon(\eta)).   
    \end{aligned}
    \end{equation}
    An infinitesimal version of \eqref{optimal cost function} is the so-called HJB equation and is written as the following form based on \eqref{Hamiltonian}
    \begin{equation} \label{HJB equation}
    0 = \min_{\upsilon \in \Psi (\Omega)} [H(\eta,\upsilon(\eta),\nabla V^{*})]. 
    \end{equation}
    Assuming that the minimum of \eqref{optimal cost function} exits and is unique. A closed form expression for the optimal control $\upsilon^*(x)$ can be derived as
    \begin{equation} \label{optimal u}
    \upsilon^*(x) = - \frac{1}{2} R^{-1} G^{\top}(\eta)\nabla V^{*}.
    \end{equation}
    Inserting \eqref{optimal u}  into the LE equation \eqref{nonlinear Lyapunov equation}, we can get the HJB equation as
    \begin{equation} \label{HJB equation expansion}
        \begin{aligned}
            0 &= \nabla V^{*}F(\eta) + P(\eta)-\frac{1}{4}\nabla {V^{*}}^{\top}G(\eta)R^{-1}G^{\top}(\eta)\nabla V^{*}.
        \end{aligned}
    \end{equation}
 \section{off policy adaptive critic learning}
The HJB equation \eqref{HJB equation expansion} is a nonlinear differential equation in terms of $\nabla V^*$. This nonlinear nature makes it extremely difficult to solve directly.
In this section, we follow the off policy adaptive critic learning technique developed in our prior work \cite{li2020online} to get the approximated solution by using the current data and experience data together. The usage of experience data to achieve weight convergence without incorporating external signals enable us to achieve PPs for full states and apply the control strategy into real applications.
\subsection{Value function approximation}
    According to the Weierstrass high-order approximation theorem \cite{finlayson2013method}, it is reasonable to conclude that there exists $W^{*}\in  \mathbb{R}^{N}$ such that the value function can be approximated as  
    \begin{equation}\label{optimal V approximation}
    V^{*}(\eta) = {W^*}^{\top} \Phi(\eta) + \epsilon(\eta),
    \end{equation}
    where $\epsilon(\eta)$ is the approximation error. Activation functions $\Phi(\eta) :  \mathbb{R}^{2n} \to \mathbb{R}^{N}$ of the artificial NN can be appropriately selected so that $V^{*}(\eta)$ and its derivative 
    \begin{equation}\label{optimal dV approximation}
    \nabla V^{*}(\eta) = \nabla \Phi^{\top}(\eta)W^{*}+\nabla\epsilon(\eta),
    \end{equation}
     are both uniformly approximated. In the above, $N$ is the number of NN activation functions. As $N \to \infty$, both $\epsilon(\eta) \to 0$ and $\nabla\epsilon(\eta) \to 0$ uniformly. Without loss of generality, the following assumption is given in this paper.
     \begin{assumption}\label{bound of NN issues}
          The approximation error of NNs is assumed to be bounded by $\left\| \epsilon(\eta)  \right\| \leq b_{\epsilon}$, and its derivative follows  $\left\| \nabla\epsilon(\eta)  \right\| \leq b_{\epsilon \eta}$. It is assumed that activation functions and their gradients are also bounded, i.e. $\left\| \Phi(\eta) \right\| \leq b_\Phi$ and $\left\| \nabla\Phi(\eta) \right\| \leq b_{\Phi \eta}$.
     \end{assumption}
    For the fixed admissible control policy $\mu(\eta)$, inserting \eqref{optimal dV approximation} into the corresponding LE \eqref{nonlinear Lyapunov equation}  yields
    \begin{equation}\label{approximation Lyapunov equation}
        {W^*}^{\top}\nabla \Phi(F(\eta)+G(\eta)\mu (\eta))+r(\eta,\mu(\eta)) = \epsilon_{h},
    \end{equation}
    where the residual error is defined as $\epsilon_{h} = -(\nabla \epsilon )^{\top}(F(\eta)+G(\eta)\mu (\eta))$. 
    Under the Lipschitz assumption on dynamics given in Assumption \ref{bound of f g}, the boundness of the residual error is denoted as $\left\| \epsilon_{h} \right\| \leq b_{\epsilon_{h}}$.
    
    Denoting $\Theta = r(\eta,\mu(\eta)) \in \mathbb{R}_{+}$ and $Y = \nabla \Phi(F(\eta)+G(\eta)\mu (\eta)) \in \mathbb{R}^{N}$, \eqref{approximation Lyapunov equation} is rewritten as 
    \begin{equation}\label{LIP Lyapunov equation}
    \Theta = -{W^*}^{\top}Y+\epsilon_{h}.
    \end{equation}
    By observing \eqref{LIP Lyapunov equation}, we know that the NN parameterized LE can be rewritten in a LIP form, which enables us to design an efficient weight estimation update law for $W^*$ with weight convergence guarantee.
    \subsection{Off policy reinforcement learning} \label{OPRL control strategy}
    Since the ideal critic weight $W^{*}$ in \eqref{LIP Lyapunov equation} is unknown, let $\hat{W}$ denote the estimated value of $W^{*}$. The output of the critic learning agent is 
    \begin{equation}\label{V approximation}
    \hat{V}(\eta) = \hat{W}^{\top} \Phi(\eta).
    \end{equation}
    The derivative of \eqref{V approximation} based on the estimated weight is
    \begin{equation}\label{approximation dV approximation}
    \nabla \hat{V}(\eta) = \nabla \Phi^{\top}(\eta)\hat{W}.
    \end{equation}
    Based on the estimated critic weight, \eqref{LIP Lyapunov equation} is rewritten as
    \begin{equation}\label{approximation LIP Lyapunov equation}
    \hat{\Theta} = -\hat{W}^{\top}Y.
    \end{equation}
    Denoting $\Tilde{W} = \hat{W}- W^{*}$, the approximation error is written as 
    \begin{equation}\label{approximation error}
    \Tilde{\Theta} = \Theta - \hat{\Theta} = \Tilde{W}^{\top}Y+\epsilon_{h}.
    \end{equation}
    Let $\hat{W}$ be adapted to minimize the  squared residual error $E = \frac{1}{2} \Tilde{\Theta}^{\top}\Tilde{\Theta}$. Then, the weight estimation update law is redesigned as 
    \begin{equation} \label{w update law}
        \dot{\hat{W}} = - \Gamma k_c Y\Tilde{\Theta} -  \sum_{l=1}^{P} \Gamma k_{e} Y_l\Tilde{\Theta}_{l},
    \end{equation}
    where $\Gamma \in \mathbb{R}^{N \times N}$ is a constant positive definite gain matrix, $k_c, k_{e} \in \mathbb{R}_{+}$ are positive constant gains to trade off the relative importance between current data and experience data to the weight estimation update law. $P \in \mathbb{R}_{+}$ is the size of the experience buffers $\mathfrak{B}$ and $\mathfrak{E}$, i.e. the maximum number of data points recorded into the experience buffers. The regression matrix $Y_{l}\in \mathbb{R}^{N}$ and the approximation error $\Tilde{\Theta}_{l}\in \mathbb{R}$ denote the $l$th collected data of the experience buffer $\mathfrak{B}$ and $\mathfrak{E}$, respectively. 
    
    In order to analyse the weight convergence problem based on the weight estimation update law in \eqref{w update law}, a rank condition is firstly clarified in Assumption~\ref{rank condition}.
    \begin{assumption} \label{rank condition}
          Given an experience buffer $\mathfrak{B} = [Y^{\top}_{1},...,Y^{\top}_{P}] \in \mathbb{R}^{N \times P}$, where $Y_l$ is the $l$th collected experience data of $\mathfrak{B}$, there holds $rank(\mathfrak{B}\mathfrak{B}^{\top}) = N$.
    \end{assumption}
    
    According to Theorem 2 provided in \cite{li2020online}, the estimated weight of the critic learning agent is guaranteed to converge to its actual value based on the weight estimation update law \eqref{w update law}. 
    The control law can be derived directly based on the estimated critic weight as
    \begin{equation} \label{approximation u}
    \hat{\mu}(\eta)= - \frac{1}{2} R^{-1} G^{\top}(\eta)\nabla \Phi^{\top}(\eta)\hat{W}.
    \end{equation}
    Finally, we can get the control applied at the dynamics \eqref{original sys} as
    \begin{equation} \label{final u}
    u = \hat{\mu} + \nu = - \frac{1}{2} R^{-1} G^{\top}(\eta)\nabla \Phi^{\top}(\eta)\hat{W} + g^{+}(x_r) (y(x_r)-f(x_r)).
    \end{equation}
    The main conclusions of this paper are given as follows.
    \begin{theorem} \label{final theorem for optimal problem }
    For the dynamics given by \eqref{new system}, the weight estimation update law is given by \eqref{w update law}, and the approximated optimal control policy is in the form of \eqref{approximation u}. Assuming that Assumption \ref{bound of f g}-\ref{rank condition} are satisfied, parameters are chosen as details in the proof. If the number of activation functions is sufficiently large, the following properties holds:
    
    (i) The approximated control policy \eqref{approximation u} stabilizes the system \eqref{new system}, and the critic weight estimation error $\Tilde{W}$ are UUB.
    The prescribed performance \eqref{Performance function} for full states of the the system \eqref{original sys} achieves under the control policy \eqref{final u} during the tracking process.
    
    (ii) The approximated optimal control $\hat{\mu}$ in \eqref{approximation u} converges to a small neighbourhood around the optimal control policy \eqref{optimal u} with the bound $\left\| \hat{\mu}-\mu\right\| \leq \epsilon_{u}$ given in \eqref{u-u* abs}.
    \end{theorem}
\begin{proof} Proof of (i). 
    Considering the following Lyapunov function candidate
    \begin{equation} \label{Lya function stability}
        V = V^{*}(\eta) + \frac{1}{2} \Tilde{W}^{\top} \Gamma^{-1} \Tilde{W}. 
    \end{equation}
    Taking time derivative of \eqref{Lya function stability} along the system \eqref{new system} yields
        \begin{equation} \label{stability dV}
        \dot{V} = \dot{V}^{*}(\eta) + \Tilde{W}^{\top} \Gamma^{-1} \dot{\hat{W}} = \dot{L}_{v}+\dot{L}_{w}.
        \end{equation}
        As for the first term $\dot{L}_{v}$
        \begin{equation} \label{stability dLv}
        \begin{aligned}
            \dot{L}_{v} &= \nabla {V^*}^{\top}(F(\eta)+G(\eta)\hat{\mu}) \\ 
            &={W^{*}}^{\top}\nabla \Phi F(\eta) - \frac{1}{2}{W^{*}}^{\top}\nabla \Phi G(\eta)R^{-1}G^{\top}(\eta)\nabla \Phi^{\top}\hat{W}\\
            &+\nabla\epsilon(\eta)(F(\eta)-\frac{1}{2} G(\eta)R^{-1}G^{\top}(\eta)\nabla \Phi^{\top}\hat{W}.
        \end{aligned}
    \end{equation}
    For simplicity, denoting $\mathcal{G} = \nabla \Phi G(\eta)R^{-1}G^{\top}(\eta)\nabla \Phi^{\top}$, it is assumed to be bounded as $\left\|\mathcal{G}\right\| \leq b_{\mathcal{G}} = b^{2}_{\Phi \eta} b^{2}_{G}/ \left\| R^{-1}\right\|$; Let $\epsilon_1 = \nabla\epsilon(\eta)(F(\eta)-\frac{1}{2} G(\eta)R^{-1}G^{\top}(\eta)\nabla \Phi^{\top}\hat{W})$ that is bounded as $\left\|\epsilon_1\right\| \leq b_{\epsilon_1}$.
    Then, \eqref{stability dLv} is rewritten as 
        \begin{equation} \label{stability dLv 2}
        \begin{aligned}
            \dot{L}_{v}  ={W^{*}}^{\top}\nabla \Phi F(\eta) - \frac{1}{2}{W^{*}}^{\top} \mathcal{G}W^{*}-
            \frac{1}{2}{W^{*}}^{\top} \mathcal{G} \Tilde{W} + \epsilon_1.
        \end{aligned}
    \end{equation}
    According to \eqref{approximation Lyapunov equation}, the following equation establishes
    \begin{equation} \label{stability dLv 3}
        \begin{aligned}
            &{W^{*}}^{\top}\nabla \Phi F(\eta) - \frac{1}{2}{W^{*}}^{\top} \mathcal{G}W^{*}\\
            &= -r(\eta,\mu(\eta)) + \epsilon_{h}\\
            &= -P(\eta) - \mu^{\top}R\mu+\epsilon_{h}\\
            & = -P(\eta)-\frac{1}{4}{W^{*}}^{\top}\mathcal{G}W^{*}+ \epsilon_{h}.
        \end{aligned}
    \end{equation}
    Finally, we can get 
        \begin{equation} \label{stability dLv 4}
            \dot{L}_{v} = -P(\eta) -\frac{1}{4}{W^{*}}^{\top}\mathcal{G}W^{*}-\frac{1}{2}{W^{*}}^{\top} \mathcal{G}\Tilde{W}+\epsilon_{h}+\epsilon_{1}.
    \end{equation}
    As for the second term $\dot{L}_{w}$, based on \eqref{w update law},
    \begin{equation} \label{stability dLw}
        \begin{aligned}
            \dot{L}_{w} &= \Tilde{W}^{\top} \Gamma^{-1}(- \Gamma k_c Y\Tilde{\Theta} - \Gamma  \sum_{l=1}^{P}k_{e}Y_{l}\Tilde{\Theta}_{l})\\
                &= -k_c \Tilde{W}^{\top}Y\Tilde{\Theta} -\Tilde{W}^{\top} \sum_{l=1}^{P}k_{e} Y_{l}\Tilde{\Theta}_{l}\\
                & = -k_c \Tilde{W}^{\top}Y(\Tilde{W}^{\top}Y+\epsilon_{h}) - \Tilde{W}^{\top} \sum_{l=1}^{P} k_{e} Y_{l}(\Tilde{W}^{\top}Y_{l}+\epsilon_{h_{l}})\\
                &=  - k_c \Tilde{W}^{\top} YY^{\top} \Tilde{W} - \Tilde{W}^{\top} \sum_{l=1}^{P}k_{e} Y_{l}Y^{\top}_{l}\Tilde{W} \\
                & +\Tilde{W}^{\top}(-k_c Y\epsilon_{h}-\sum_{l=1}^{P}k_{e} Y_{l}\epsilon_{h_{l}})\\
                &\leq - \Tilde{W}^{\top}\sum_{l=1}^{P}k_{e} Y_{l}Y^{\top}_l\Tilde{W} +\Tilde{W}^{\top}(-k_c Y\epsilon_{h}-\sum_{l=1}^{P}k_{e} Y_{l}\epsilon_{h_{l}}).\\
                &= - \Tilde{W}^{\top}X\Tilde{W}-\Tilde{W}^{\top}\epsilon_{er},
        \end{aligned}
    \end{equation}
    where $X= \sum_{l=1}^{P}k_{e} Y_{l}Y^{\top}_l$, $\epsilon_{er} =k_c Y^{\top}\epsilon_{h}+\sum_{l=1}^{P}k_{e} Y^{\top}_l \epsilon_{h_{l}}$ which is bounded by $\left\|\epsilon_{er}\right\| \leq b_{\epsilon_{er}}$.
    
    Finally, substituting \eqref{stability dLv 4} and \eqref{stability dLw} into \eqref{stability dV}, based on the fact that $\left\|W^{*}\right\| \leq b_{W^{*}}$, we can get
        \begin{equation} \label{stability dLv final}
        \begin{aligned}
            \dot{L}_{v} &= -P(\eta) -\frac{1}{4}{W^{*}}^{\top}\mathcal{G}W^{*}-\Tilde{W}^{\top}X\Tilde{W}\\
            &+\Tilde{W}(-\epsilon_{er}-\frac{1}{2}{W^{*}}^{\top} \mathcal{G})+\epsilon_{h}+\epsilon_{1} \\
            & \leq -P(\eta) -\frac{1}{4}{W^{*}}^{\top}\mathcal{G}W^{*} - \lambda_{min}(X)\left\|\Tilde{W}\right\|^2\\
            &+(b_{\epsilon_{er}}+1/2b_{\mathcal{G}}b_{W^{*}})\left\|\Tilde{W}\right\|+b_{\epsilon_{h}}+b_{\epsilon_{1}}\\
            & = -\mathcal{A}-\mathcal{B}\left\| \Tilde{W}\right\|^2+\mathcal{C}\left\| \Tilde{W}\right\|+\mathcal{D},
        \end{aligned}
    \end{equation}
    where $\mathcal{A} = P(\eta) +\frac{1}{4}{W^{*}}^{\top}\mathcal{G}W^{*}$ is positive, $\mathcal{B}  = \lambda_{min}(X)$, $\mathcal{C} = b_{\epsilon_{er}}+1/2b_{\mathcal{G}}b_{W^{*}} $ and $\mathcal{D} = b_{\epsilon_{h}}+b_{\epsilon_{1}}$.
    
    Since $\mathcal{A}$ is positive definite, the above Lyapunov derivative is negative if
    \begin{equation} \label{negative condition}
        \begin{aligned}
        \left\| \Tilde{W} \right\| > \frac{\mathcal{C}}{2\mathcal{B}}+\sqrt{\frac{\mathcal{C}^2}{4\mathcal{B}^2}+\frac{\mathcal{D} }{\mathcal{B}}}.
        \end{aligned}
    \end{equation}
    
    Thus, the critic weight estimation error converges to the residual set defined as
    \begin{equation} \label{compact set}
        \begin{aligned}
        \Tilde{\Omega}_{\Tilde{W}} = \{\Tilde{W} | \left\| \Tilde{W} \right\| \leq \frac{\mathcal{C}}{2\mathcal{B}}+\sqrt{\frac{\mathcal{C}^2}{4\mathcal{B}^2}+\frac{\mathcal{D} }{\mathcal{B}}} \}.
        \end{aligned}
    \end{equation}
    Denoting $V(0)$ as the value of the Lyapunov function candidate $V$ at $t = 0$, it is a bounded function determined by initial values. According to the above derivation, $\dot{V} <0$ establishes, which means that $\forall t$, $V(t) < V(0)$ always establishes, i.e. $V(t)$ is a bounded function at any time. The boundness of $V(t)$ implies that prescribed performance related constraints will not be violated. Otherwise, $V(t) \to \infty$ if any  constraint violation happens. Thus, we can conclude that the prescribed performance of full states achieves.
    
     Proof of (ii).  The difference between the approximated optimal control and optimal control follows
    \begin{equation}\label{u-u* abs}
        \begin{aligned}
          &\left\|\hat{\mu}(\eta)-\mu(\eta) \right\| \\
         & \leq\left\|-\frac{1}{2} R^{-1} G^{\top}(\eta)\nabla \Phi^{\top}(\eta)\tilde{W} + \frac{1}{2} R^{-1} G^{\top}(\eta)\nabla \Phi^{\top} \nabla \epsilon \right\| \\
          & \leq -\frac{1}{2} b_{G}b_{\nabla \eta}\left\|R^{-1}\right\| \left\|\Tilde{W}\right\|+\frac{1}{2}\left\|R^{-1}\right\| b_{G}b_{\Phi \eta}b_{ \epsilon \eta}=\epsilon_{u}.
        \end{aligned}
    \end{equation}
\end{proof}

\section{Simulation results}
A 2-DoF robot manipulator is chosen to show the effectiveness of the proposed control method. The Euler-Lagrange (E-L) model is given as
\begin{equation}\label{2DoF model}
M\Ddot{q}+C\dot{q}+F_d\dot{q}+F_s = \tau.
\end{equation}
where $q \in \mathbb{R}^2$, $\dot{q} \in \mathbb{R}^2$ and $\Ddot{q} \in \mathbb{R}^2$ are the vectors of joint angles, velocities, and accelerations respectively; $M=\begin{bmatrix}
   m_{11}  & m_{12} \\
    m_{12} & m_{22} \\
\end{bmatrix} $ is the inertia matrix with $m_{11}=p_{1}+2p_{3}c_2$,  $m_{12} =p_{2}+p_{3}c_2$, and $m_{22} = p_{2}$; $C = \begin{bmatrix}
   c_{11}  & c_{12} \\
    c_{21} & c_{22} \\
\end{bmatrix}$ is the matrix of centrifugal and Coriolis terms with $c_{11} = -p_3 s_2 \dot{q}_2$, $c_{12} = -p_3 s_2(\dot{q}_1+\dot{q}_2)$, $c_{21} = p_3 s_2 \dot{q}_1$, and $c_{22}=0$; $F_d = \diag{f_{d1},f_{d2}}$ stands for the dynamic friction; $F_s = [f_{s1}\tanh{(\dot{q}_1)},f_{s2}\tanh(\dot{q}_2)]^{\top}$ denotes the static friction. The explicit values for the robot dynamics are set as $p_1 = 3.4743$, $p_2 = 0.196$, $p_3 = 0.242$, $f_{s1} = 8.45$, $f_{s2} = 2.35$, $f_{d1} = 5.3$, $f_{d2} = 1.1$.
The E-L equation \eqref{2DoF model} can be written in the form of \eqref{original sys} by setting $x=[x_1,x_2,x_3,x_4]=[q_1,q_2,\dot{q}_1,\dot{q}_2]^{\top}$,
$f(x) = [x_3,x_4,(M^{-1}(-C-F_d)[x_3, x_4]^{\top}-F_s)^{\top}]^{\top}$, and $g(x) = [[0,0]^{\top},[0,0]^{\top},(M^{-1})^{\top}]^{\top}$. 
For simulation, the reference trajectory is set as $x_r(t) = [0.5\cos{(2t)},\cos{t},-\sin{(2t)},-\sin{(t)}]^{\top}$. Then, we get $\dot{x}_r = y(x_r) = [-\sin{(2t)},-\sin{(t)},-2\cos{(2t)},-\cos{(t)}]$, $g_r^{+}=[[0,0]^{\top},[0,0]^{\top},M^{\top}(x_r)]^{\top}$.
\subsection{OTCP case}\label{Simulation case OPT}
In this section, the effectiveness of the proposed off policy adaptive critic learning architecture illustrated in Section \ref{OPRL control strategy} is tested to tackle the common OTCP. The robot manipulator \eqref{2DoF model} is driven to track the reference trajectory $x_r$ while minimizing the common quadratic cost function 
\begin{equation}\label{common cost function optimal tracking}
    V(\eta) = \int_{0}^{\infty} e^{\top} Q e  + \mu^{\top}R\mu \,dt,
\end{equation}
where $Q = \diag{[8,8,8,8]}$, $R =1$. The basis set $\Phi(\eta) \in \mathbb{R}^{23}$  is chosen as
\begin{equation}\label{basise set 2DoF}
\begin{aligned}
        \Phi(\eta) = \frac{1}{2} & [\eta^2_1,\eta^2_2,2\eta_1\eta_3,2\eta_1\eta_4,2\eta_2\eta_3,2\eta_2\eta_4,\eta^2_1\eta^2_2,\eta^2_1\eta^2_5,\\
    &\eta^2_1\eta^2_6,\eta^2_1\eta^2_7,\eta^2_1\eta^2_8,\eta^2_2\eta^2_5,\eta^2_2\eta^2_6,\eta^2_2\eta^2_7,\eta^2_2\eta^2_8,\eta^2_3\eta^2_5,\\  
    &\eta^2_3\eta^2_6,\eta^2_3\eta^2_7,\eta^2_3\eta^2_8,\eta^2_4\eta^2_5,\eta^2_4\eta^2_6,\eta^2_4\eta^2_7,\eta^2_4\eta^2_8]^{\top}.
\end{aligned}
\end{equation}
The size of the experience buffer is set as $P = 25$. For the weight estimation update law \eqref{w update law}, parameters are set as $k_e = 10$, $k_c = 100$, and $\Gamma = I_{23 \times 23}$. For simulation, the initial values are set as $x_0 = [0.4,1.1,0,0]^{\top}$, $\hat{W} = 0_{1\times 23}$. Simulation results for this typical OTCP is shown from Fig.\ref{fig of critic weight OPT} to Fig.\ref{fig of tracking error OPT}.

The critic weight convergence result is shown in Fig.\ref{fig of critic weight OPT}. We know that after $t=40 \mathrm{s}$, the convergence of the estimated critic weight $\hat{W}$ achieves without incorporating probing noises. 
\begin{figure}[!t]
    \centering
    \includegraphics[width=3.8in]{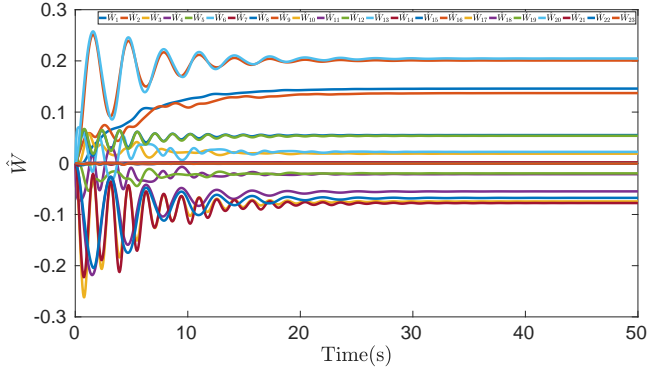}
    \caption{The weight convergence result of the critic learning agent $\hat{W}$ for the OTCP case.}
    \label{fig of critic weight OPT}
\end{figure}
The trajectories of full states and their references are displayed in Fig.\ref{state traj OPT}, and the trajectory of the tracking error $e$ is shown in Fig.\ref{fig of tracking error OPT}. It is concluded that the proposed control strategy enable the robot manipulator track the reference trajectory precisely.
\begin{figure*}[!t]
\centering
\subfloat[Trajectories of of $q_{1}$ and its reference $q_{1r}$]{\includegraphics[width=3.7in]{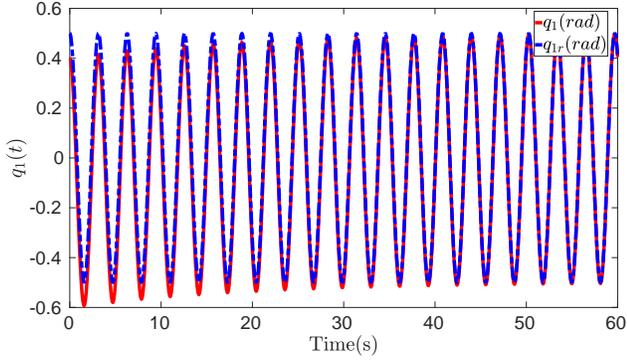}%
\label{q1 trajectory OPT}}
\subfloat[Trajectories of $q_{2}$ and its reference $q_{2r}$ ]{\includegraphics[width=3.7in]{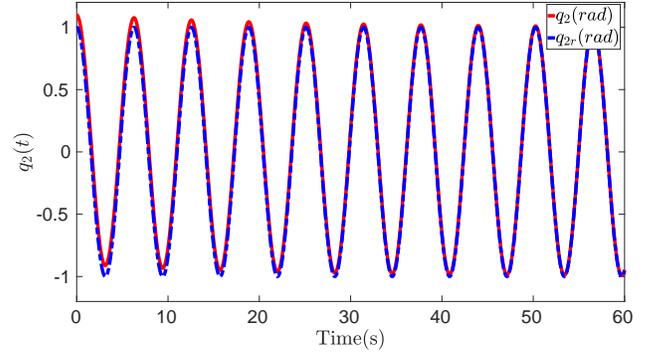}%
\label{q2 trajectory OPT}}

\subfloat[Trajectories of  $\dot{q}_{1}$ and its reference $\dot{q}_{1r}$ ]{\includegraphics[width=3.7in]{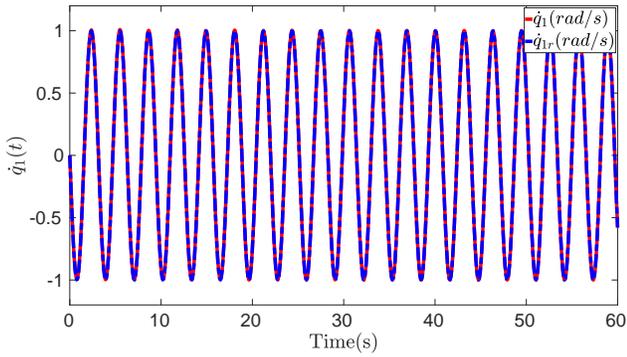}%
\label{qd1 trajectory OPT}}
\subfloat[Trajectories of of $\dot{q}_{2}$ and its reference $\dot{q}_{2r}$ ]{\includegraphics[width=3.7in]{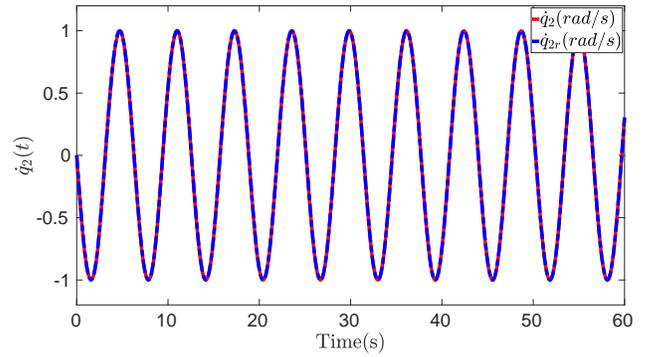}%
\label{state trajectory OPT}}
\caption{The trajectories of full states and their references for the OTCP case.}
\label{state traj OPT}
\end{figure*}
\begin{figure}[!t]
    \centering
    \includegraphics[width=3.8in]{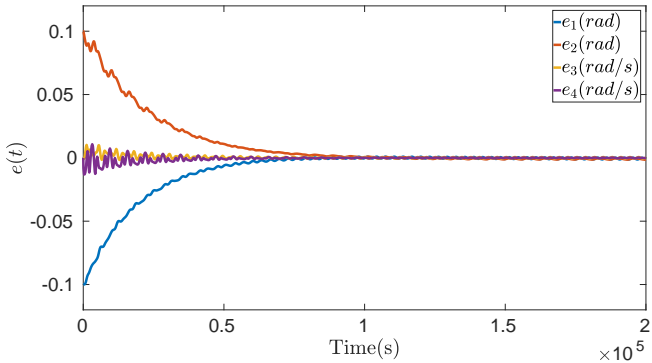}
    \caption{The trajectories of tracking errors for the OTCP case}
    \label{fig of tracking error OPT}
\end{figure}
\subsection{PP-OTCP case}
In this part, the effectiveness of off policy RL based method to solve the PP-OTCP illustrated in Problem \ref{PPTC} is tested. For simulation, PPFs for full states are set as 
\begin{equation} \label{PF sim}
\begin{aligned}
     \rho_i(t) = (60\pi/180 - 3\pi/180) e^{-0.1t}+3\pi/180, i = 1,2,3,4
\end{aligned}
\end{equation}
To achieve optimal trajectory tracking control with PPs, the cost function is designed as
\begin{equation}\label{PP cost fuction}
    V(\eta) = \int_{t}^{\infty} \sum_{i=1}^{4} k_{i} \log \frac{\alpha^2_{i}}{\alpha^2_{i}-\zeta^2_{i}}+h_{i} \log \frac{\beta^2_{i}}{\beta^2_{i}-\delta^2_{i}} +\mu^{\top}R\mu \,d\tau
\end{equation}
where $k_1 = 1$, $\alpha_1 = 0.20$; $k_2 = 0.3$, $\alpha_2 = 0.25$; $k_3 = 1$, $\alpha_3 = 0.25$;  $k_4 = 1$, $\alpha_4 = 0.25$; $h_i =0.01,\beta_i=10, i =1,2,3,4$.

For a fair comparison, the parameters are set as same with the OTCP case in Section \ref{Simulation case OPT}.
The parameter convergence result is shown in Fig.\ref{fig of critic weight PPT}. After 50 seconds, parameter convergence result achieves.
\begin{figure}[!t]
    \centering
    \includegraphics[width=3.8in]{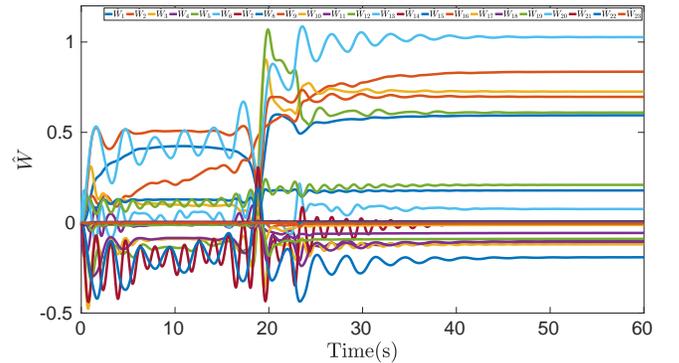}
    \caption{The weight convergence result of the critic learning agent $\hat{W}$ for the PP-OTCP case.}
    \label{fig of critic weight PPT}
\end{figure}
In order to show the effectiveness of the proposed method to achieve PPs, the comparison results are displayed from Fig.\ref{fig of tracking error e1 PPT} to  Fig.\ref{fig of tracking error e4 PPT}. As shown in Fig.\ref{fig of tracking error e1 PPT} and Fig.\ref{fig of tracking error e2 PPT}, the trajectories of $e_1$ and $e_2$ based on the common quadratic cost function \eqref{common cost function optimal tracking} violate the boundaries of PPFs, while our proposed method can effectively drive the robot manipulator to track reference trajectory and satisfy the performance requirements defined by PPFs.
\begin{figure}[!t]
    \centering
    \includegraphics[width=3.8in]{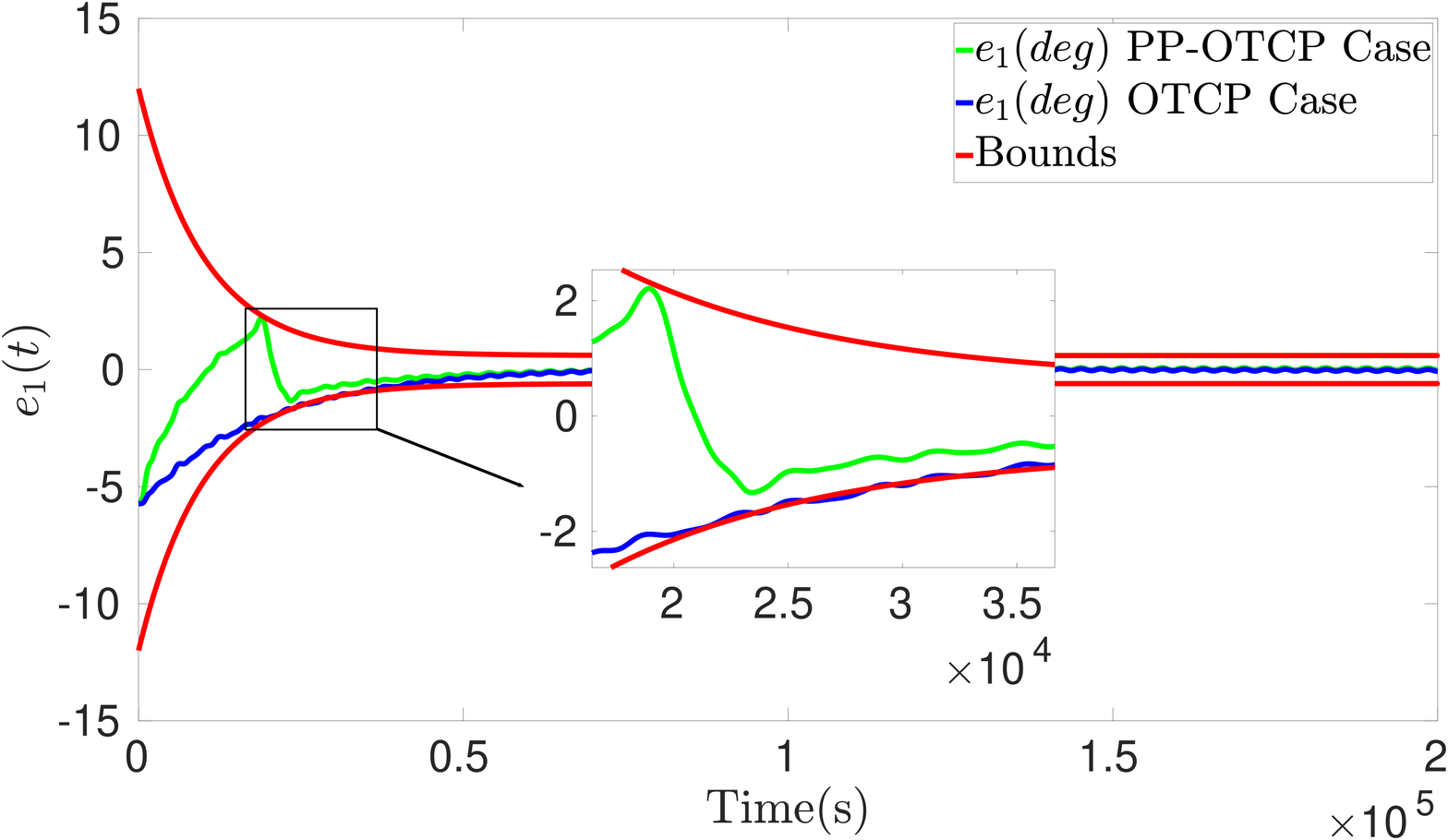}
    \caption{Comparison results of the tracking error $e_1$ }
    \label{fig of tracking error e1 PPT}
\end{figure}
\begin{figure}[!t]
    \centering
    \includegraphics[width=3.8in]{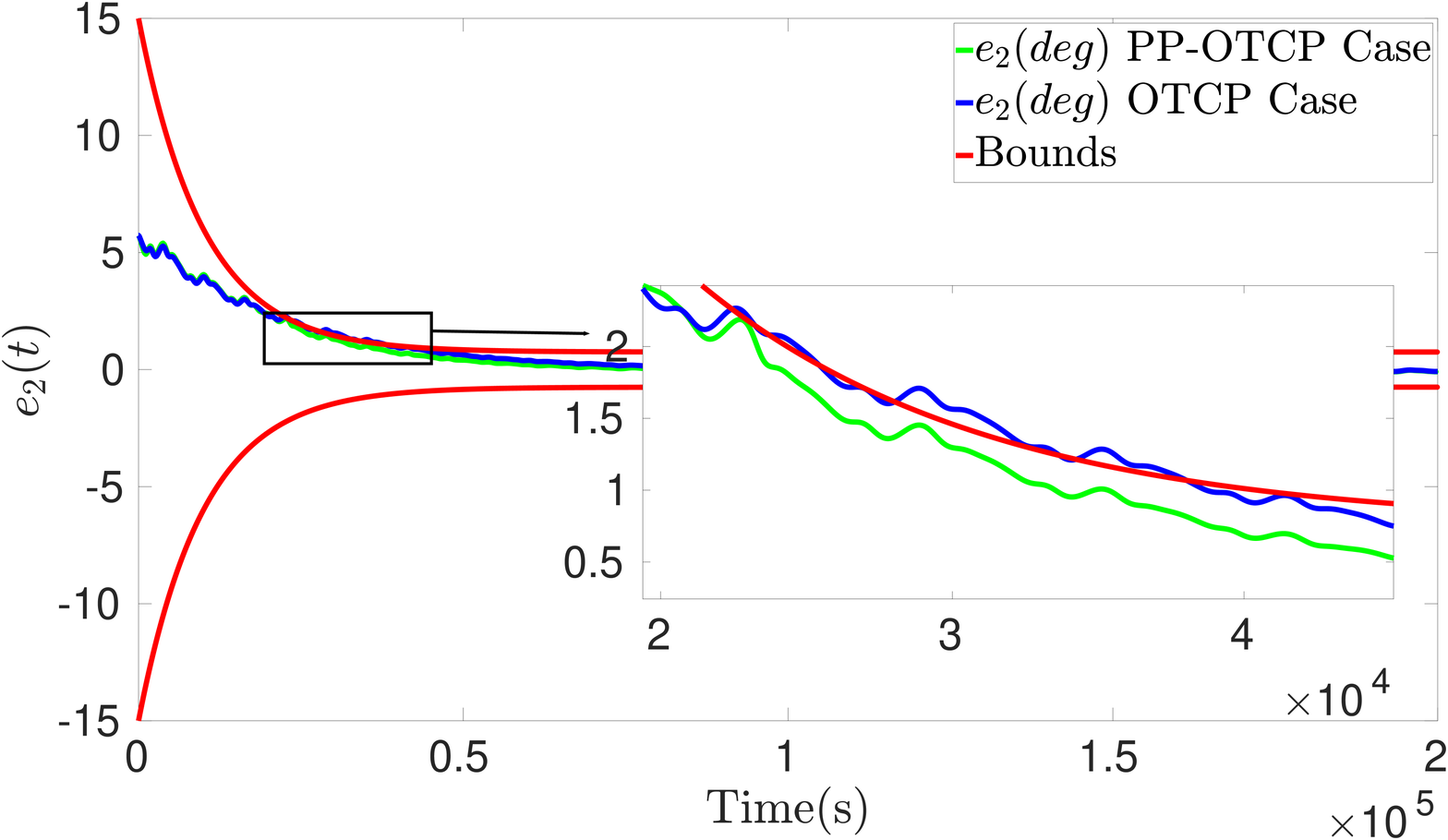}
    \caption{Comparison results of the tracking error $e_2$ }
    \label{fig of tracking error e2 PPT}
\end{figure}
\begin{figure}[!t]
    \centering
    \includegraphics[width=3.8in]{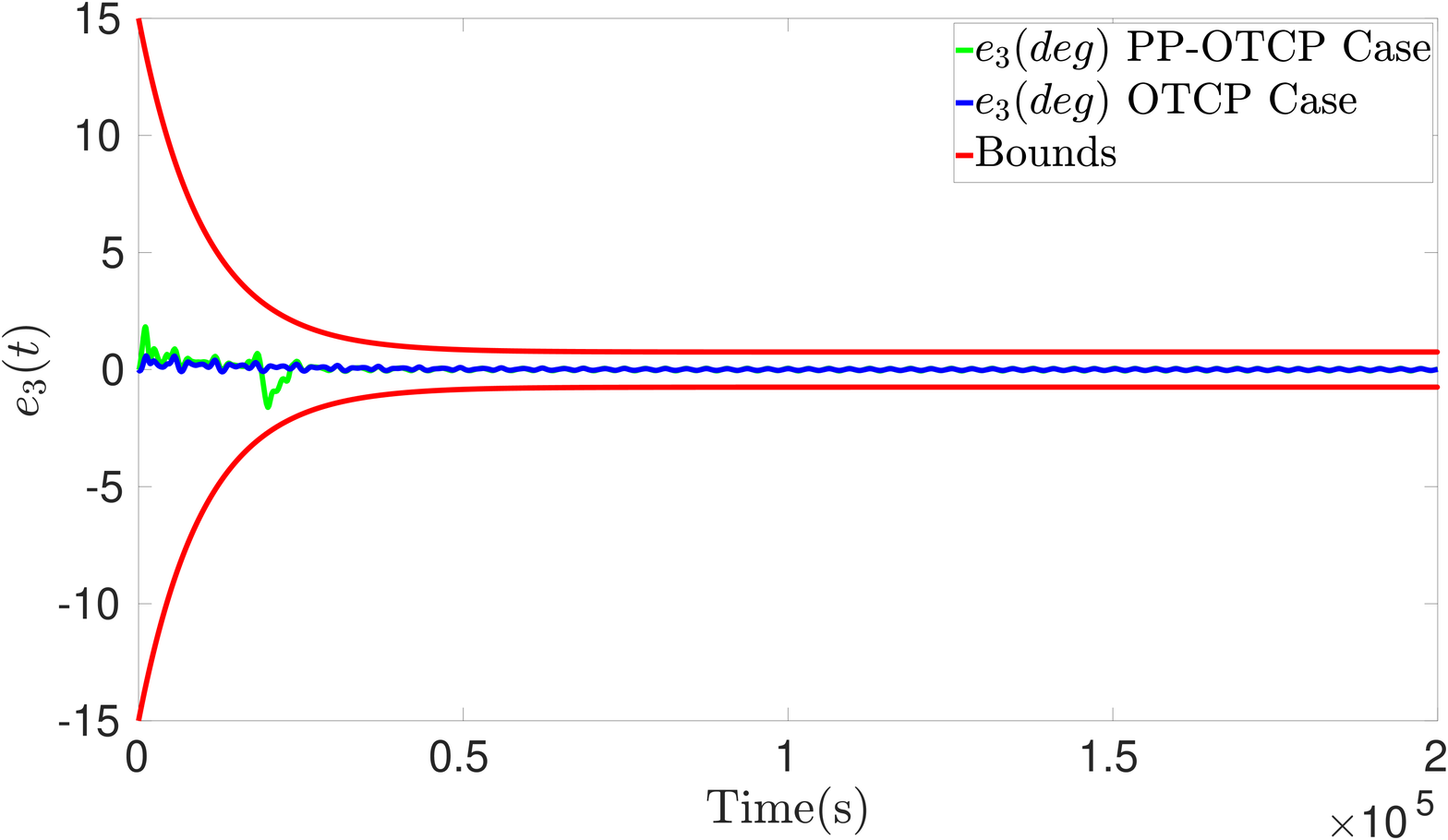}
    \caption{Comparison results of the tracking error $e_3$ }
    \label{fig of tracking error e3 PPT}
\end{figure}
\begin{figure}[!t]
    \centering
    \includegraphics[width=3.8in]{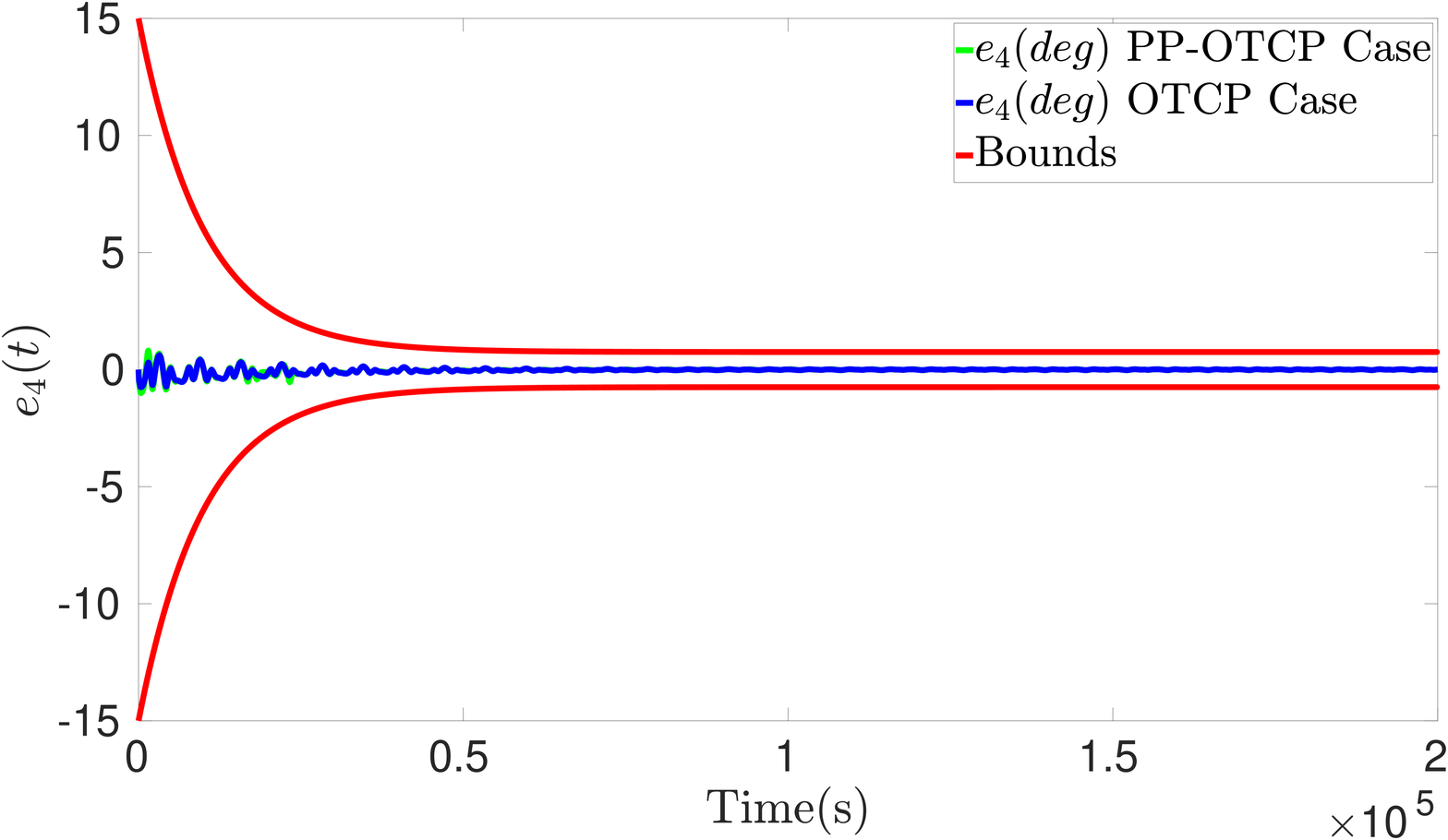}
    \caption{Comparison results of the tracking error $e_4$ }
    \label{fig of tracking error e4 PPT}
\end{figure}
\section{Conclusion}
To achieve prescribed performances for full states of the optimal tracking control problem, an off policy risk sensitive RL based control strategy is developed in this paper.
An auxiliary system is proposed to transform the optimal tracking control problem as an optimal regulation problem. 
The required prescribed performances are reflected by risk sensitive state penalty terms that are incorporated into the cost function of the transformed optimal regulation problem.
The HJB equation is approximately solved based on an off policy adaptive critic learning architecture, which achieves weight convergence without incorporating external signals to satisfy the PE condition. Simulation results have proved the effectiveness of the proposed strategy. In the future, experiments will be conducted to show the effectiveness of the proposed strategy on a 3-DoF robot manipulator. 

\bibliographystyle{IEEEtran}
\bibliography{bibtex/bib/IEEEexample}

\end{document}